\title {An FKN Theorem for the Binary Grassmann Scheme}
\author{
	Yuval Filmus\thanks{Taub Faculty of Computer Science and Faculty of Mathematics, Technion Israel Institute of Technology. The work was supported by the Israel Science Foundation, Grant No.\ 507/24.}
	\and
    Anqi Li\thanks{Department of Mathematics, Stanford University.}
    \and 
	Dor Minzer\thanks{Department of Mathematics, Massachusetts Institute of Technology. Supported by NSF CCF award 2227876 and NSF CAREER award 2239160.}}
\date{\vspace{-5ex}}

\documentclass[11pt]{article}
\usepackage{times}
\usepackage[T1]{fontenc}
\usepackage{amssymb}
\usepackage{amsmath}
\usepackage{amsthm}
\usepackage{thmtools}
\usepackage{thm-restate}
\usepackage{bm}
\usepackage{xcolor}
\usepackage{fullpage}

\usepackage{hyperref}
\hypersetup{hypertexnames=false}

\usepackage{thmtools}
\usepackage{thm-restate}

   \newtheorem{thm}{Theorem}[section]
   
   \newtheorem{lemma}[thm]{Lemma}
   
   \newtheorem{claim}[thm]{Claim}
   \newtheorem{fact}[thm]{Fact}
   \newtheorem{remark}[thm]{Remark}
   \newtheorem{definition}[thm]{Definition}

   \newtheorem{problem}{Problem}

\newcommand\E{\mathbb{E}}
\newcommand\card[1]{\left| {#1} \right|}

\newcommand\sett[2]{\left\{ \left. #1 \;\right\vert #2 \right\}}

\newcommand\Prob[2]{{\Pr_{#1}\left[ {#2} \right]}}

\newcommand\norm[1]{\| #1 \|}

\newcommand\Expect[2]{{\mathop{\mathbb{E}}_{#1}\left[ {#2} \right]}}

\newcommand\inner[2]{\langle{#1},{#2}\rangle}
\newcommand\eps{\varepsilon}

\renewcommand\geq{\geqslant}
\renewcommand\leq{\leqslant}

\newcommand{\rom}[1]{\uppercase\expandafter{\romannumeral #1\relax}}

\providecommand{\sqbinom}[2]{\genfrac{[}{]}{0pt}{}{#1}{#2}}

\begin{document}

\maketitle
\begin{abstract}
A classical theorem due to Friedgut, Kalai and Naor asserts that if a function $f\colon \{0,1\}^n\to\{-1,1\}$ close to a degree $1$ function, then either $f$ or $-f$ is close to either the all $1$ function, or to $(-1)^{x_i}$ for some $i\in [n]$. We prove a version of their theorem for the Grassmann scheme over $\mathbb{F}_2$. More precisely, we prove if a function $f\colon \sqbinom{\mathbb{F}_2^n}{\ell}\to\{0,1\}$ is close to a degree $1$ function, then either $f$ or $1-f$ must be close to a function of the form $g(L) = \sum\limits_{x\in\mathcal{X}}1_{x\in L}+\sum\limits_{W\in\mathcal{W}}1_{L\subseteq W}$, where $\mathcal{X}\subseteq\mathbb{F}_2^n$ is a set of points and $\mathcal{W}$ is a set of hyperplanes in $\mathbb{F}_2^n$.
\end{abstract}
\section{Introduction}
The classical Friedgut--Kalai--Naor (FKN) Theorem~\cite{friedgut2002boolean} asserts that if a Boolean function $f\colon \{0,1\}^n\to\{-1,1\}$ is close to a degree $1$ function, then either $f$ or $-f$ must be close either to the all $1$ function, or to a dictatorship function. More precisely, if the Fourier expansion $f(x) = \sum\limits_{\alpha\in\mathbb{F}_2^n}\widehat{f}(\alpha)(-1)^{\inner{\alpha}{x}}$ satisfies that 
\[
\sum\limits_{\card{\alpha}\leq 1} \widehat{f}(\alpha)^2\geq 1-\eps,
\]
then there exists $\alpha\in\mathbb{F}_2^n$ such that $\lvert\widehat{f}(\alpha)\rvert\geq 1-O(\eps)$. This result belongs to the broader class of junta theorems, which are results transferring information regarding the structure of a function in the Fourier domain to information regarding the structure of the function in physical space. Notable examples include the Kindler--Safra theorem~\cite{KS}, the Bourgain junta theorem~\cite{Bourgain2002FourierSpectrum}, and their extensions and refinements~\cite{AlonDinurFriedgutSudakov2004,DinurFriedgutKindlerODonnell2007,JendrejOleszkiewiczWojtaszczyk2012,RubinsteinSafra2015,EFF1,EFF2,EFF3,Filmus2020FKNMultislice,Filmus2021BooleanSn,EldanKindlerLifshitzMinzer2025}. Most relevant to the current paper are the works~\cite{FilmusIhringer,Ihringer2023Classification,Filmus2026Grassmann}, which study the problem over the \emph{Grassmann scheme}.

For a linear space $W$ over $\mathbb{F}_2$ and $\ell\leq \dim(W)$, the vertices of the Grassmann graph $\mathsf{Gr}(W,\ell)$ are $\sqbinom{W}{\ell}$, the set of linear subspaces of $W$ of dimension $\ell$, and $L,L'\in {\sqbinom{W}{\ell}}$ are adjacent if $\dim(L\cap L')=\ell-1$. This graph has recently received much attention in the theoretical computer science community due to its relation to the proof of the $2$-to-$1$-Games Conjecture with imperfect completeness~\cite{KMS1,DKKMS1,DKKMS2,KMS2}. It can be viewed as the sparsest member of the family of short code graphs, which sometimes can be used as more size efficient versions of the long-code~\cite{barak2015making,khot2017hardness}.
The main objective of the current paper is to establish an FKN-style theorem for the Grassmann scheme (we remark that weaker versions of it which have already found some applications~\cite{KaufmanM,MZheng}).

\subsection{The Degree Decomposition} 
To state the problem we study precisely and state our result, we first present the level decomposition on the Grassmann scheme. 
\begin{definition}
    For an integer $0\leq i\leq \ell$, the space $J_{\leq i}$ of degree at most $i$ functions is the span of the functions $\{f_I\}_{I\subseteq W\text{ subspace}, \dim(I)=i}$ where $f_I(L) = 1_{I\subseteq L}$. We define the operator $P_{\leq i}\colon L_2({\sqbinom{W}{\ell}})\to L_2({\sqbinom{W}{\ell}})$ to be the orthogonal projection onto $J_{\leq i}$.
\end{definition}

In these notations, the space of degree $1$ functions is simply $J_{\leq 1}$. To measure distances, we equip the set of real-valued functions with the inner product and $L_p$ norms defined as
\[
\inner{f}{g}=\Expect{L\in {\sqbinom{W}{\ell}}}{f(L)g(L)},
\qquad
\norm{f}_p=\left(\Expect{L\in{\sqbinom{W}{\ell}}}{\card{f(L)}^p}\right)^{1/p},
\]
where the expectation is taken with respect to the uniform distribution. We also use the notation $\mathsf{var}(f) = \Expect{L\in {\sqbinom{W}{\ell}}}{\card{f(L)-\E[f]}^2} = \E[f^2]-\E[f]^2$ for a real-valued function $f$.
The main problem we study is the following:
\begin{problem}\label{problem:main}
    Suppose that a Boolean function $f\colon \sqbinom{\mathbb{F}_2^n}{\ell}\to\{0,1\}$ satisfies that $\norm{f-g}_2^2\leq \eps\mathsf{var}(f)$ for some $g\in J_{\leq 1}$, where $\eps>0$ is thought of as small. What can we say about the structure of $f$?
\end{problem}
Prior works~\cite{FilmusIhringer,Ihringer2023Classification,Filmus2026Grassmann} studied the exact classification problem, namely the case that $\eps=0$. They established the following result:
\begin{thm}\label{thm:exact}
    Suppose that $\ell,n-\ell\geq 2$, and suppose that $f\colon {\sqbinom{\mathbb{F}_2^n}{\ell}}\to\{0,1\}$ is from $J_{\leq 1}$. Then either $f$ or $1-f$ is a function of the form
    \[
    0,\qquad 1_{x\in L},\qquad 1_{L\subseteq W},\qquad 1_{x\in L}+1_{L\subseteq W},
    \]
    for a point $x$ and a hyperplane $W$, and in the last case we have that $x\not\in W$.
\end{thm}
Our main result, stated in the next section, gives a robust version of this statement in the case that $\ell,n-\ell$ are both thought of as large.

\subsection{Main Result}
Our main result is the following assertion.
\begin{thm}\label{thm:main}
    For all $\delta>0$ there exist $\eps>0$ and $T\in\mathbb{N}$ such that the following holds for $\ell<n$ such that $\ell,n-\ell\geq T$. Suppose a function $f\colon {\sqbinom{\mathbb{F}_2^n}{\ell}}\to\{0,1\}$ satisfies that 
    $\norm{f-P_{\leq 1} f}_2^2\leq \eps\mathsf{var}(f)$. Then there exist a set of points $\mathcal{X}\subseteq \mathbb{F}_2^n$ and a set of hyperplanes $\mathcal{W}$ such that defining
    \[
    g(L) = \sum\limits_{x\in\mathcal{X}}1_{x\in L}
    +
    \sum\limits_{W\in\mathcal{W}}1_{L\subseteq W}, 
    \]
    we have that $\norm{f-g}_2^2\leq \delta \mathsf{var}(f)$, or 
    $\norm{(1-f)-g}_2^2\leq \delta\mathsf{var}(f)$. Also, we have that $\mathsf{var}(f)\leq \delta$ and 
    $\norm{g - G}_2^2 \leq O(\mathsf{var}(f)^2)$, where $G(L)$ is obtained from $g(L)$ by rounding to the closest Boolean value.
\end{thm}
While we did not explicitly compute the dependency of $\eps, T$ on $\delta$, it is worth mentioning that $\eps$ is much smaller than polynomial in $\delta$, but as far as we know the theorem may be correct even with $\delta = O(\eps)$ (but still thinking of both $\ell$ and $n-\ell$ as large). Addressing the case that one of $\ell$ and $n-\ell$ is small seems to be more challenging.

\begin{remark}
The difference in structure between Theorem~\ref{thm:exact} and Theorem~\ref{thm:main} resembles a phenomenon encountered in other domains such as the $p$-biased cube and the symmetric group. Taking as an example the $p$-biased cube, if $f\colon \{0,1\}^n \to \{0,1\}$ satisfies $\|f^{>1}\|^2 \le \eps$ (with respect to $\mu_p$) then
\begin{itemize}
    \item \textbf{Coarse approximation:} $f$ or $1-f$ is $O(\sqrt{\eps} + p)$-close to a Boolean constant. In our case, this corresponds to the part of Theorem~\ref{thm:main}, asserting that $\mathsf{var}(f)$ is small.
    \item \textbf{Refined approximation:} $f$ or $1-f$ is $O(\eps)$-close to a sum of $m$ many $x_i$'s, where $m$ is upper bounded by $\max(1,O(\sqrt{\eps}/p))$. In our case, this corresponds to the part of Theorem~\ref{thm:main} asserting that $f$ and $g$ are close, even relative to $\mathsf{var}(f)$.
\end{itemize}
\end{remark}

\begin{remark}
Theorem~\ref{thm:exact} generalizes to arbitrary finite fields~\cite{Ihringer2023Classification,Filmus2026Grassmann}, with the additional assumption that $\max(\ell,n-\ell)$ is larger than a constant depending on~$q$.

Our main result Theorem~\ref{thm:main} also generalizes to arbitrary finite fields, with essentially the same proof. In order to keep the paper concise, we work exclusively over $\mathbb{F}_2$.
\end{remark}

\subsection{Proof Overview}
The proof of Theorem~\ref{thm:main} proceeds in three steps, which we outline below. A key intermediate goal is the following result:
\begin{restatable}{lemma}{MainLemma}
\label{lem:almost_constnat_absolute}
    For all $\delta>0$, there exists $\eps>0$ and $T\in\mathbb{N}$, such that if $\ell,n\in\mathbb{N}$ are integers such that $\ell,n-\ell\geq T$, and $f\colon {\sqbinom{\mathbb{F}_2^n}{\ell}}\to\{0,1\}$ satisfies that $\norm{f-g}_1\leq \eps$ for some $g\in J_{\leq 1}$, then
    \[
    \E[f]\leq \delta\qquad\text{or}\qquad \E[f]\geq 1-\delta.
    \]
\end{restatable}
The first two steps establish Lemma~\ref{lem:almost_constnat_absolute}, and in the third step we show how to derive Theorem~\ref{thm:main} from it. We now elaborate on each one of these steps.

In the first step, we establish a weak form of (a variant of) Lemma~\ref{lem:almost_constnat_absolute}, where the closeness parameter $\eps$  is allowed to depend on the dimensions $n$ and $\ell$. Such an assertion turns out to follow easily from the exact classification result of~\cite{FilmusIhringer} via  a compactness argument. 
In the second step, we lift this weak result to a corresponding result for large dimensions $n$ and $\ell$. This step proceeds by applying an appropriate analog of random restrictions on the Grassmann scheme (which was proposed to us by ChatGPT, as discussed below). Using the weak result, we are able to claim that with  probability close to $1$, random restrictions of the function $f$  are close to being constant. We conclude from it that $f$ itself must be close to $\delta$-constant, for some $\delta = \delta(\eps)$ that vanishes with $\eps$, thereby establishing Lemma~\ref{lem:almost_constnat_absolute}.

In the third step we conclude the stronger closeness guarantee of Theorem~\ref{thm:main}. Here, the main point is that the closeness between $f$ and $g$ is relative to the measure of $f$. To do so we use global hypercontractivity to first find a set of points $\mathcal{X}$ and hyperplanes $\mathcal{W}$ such that the conditional expectations 
$\Expect{L\ni x}{f(L)},\Expect{L\subseteq W}{f(L)}$ are noticeable for every $x\in\mathcal{X}$, $W\in\mathcal{W}$, and collectively they cover ``almost all'' of the $L$ such that $f(L)=1$. We then show, using conditioning and the preceding Lemma~\ref{lem:almost_constnat_absolute}, that for most $x\in\mathcal{X}$ (and similarly, for most $W\in\mathcal{W}$), the corresponding conditional expectation must actually be close to $1$.
\vspace{-2ex}
\paragraph{Statement of AI use:} The authors used ChatGPT 5.6 Pro to find the dimension reduction argument in Section~\ref{sec:dim_reduction}. The authors then prompted chatGPT 5.6 with the high-level plan of the proof as presented in this paper. Initially it raised concerns about gaps in the arguments, and the authors offered ways of filling these gaps in, and eventually it produced a proof. The version presented below was written by the authors.
\section{Preliminaries}
We will use the normalized adjacency operator $\mathrm{T}_{\mathsf{Gr}(W,\ell)}\colon L_2({\sqbinom{W}{\ell}})\to L_2({\sqbinom{W}{\ell}})$ of the Grassmann graph, which is defined by 
\[
\mathrm{T}_{\mathsf{Gr}(W,\ell)}f(L)
=\Expect{\substack{L'\in{\sqbinom{W}{\ell}}\\ \dim(L'\cap L)=\ell-1}}{f(L')}.
\]
\begin{definition}
    For $i\geq 0$, we define the space $J_{=i} = J_{\leq i}\cap J_{\leq i-1}^{\perp}$ (under the convention that $J_{=0}$ is the space of all constant functions). We define $P_{=i}$ to be the orthogonal projection onto $J_{=i}$.
\end{definition}

It is well known that the spaces $J_{=i}$ are eigenspaces of $\mathrm{T}_{\mathsf{Gr}(W,\ell)}$, and there are exact formulas for the eigenvalues corresponding to them~\cite{BCN1989,GodsilMeagher2016}. In particular, the second eigenvalue of this operator is at most $1/2$ (see~\cite[Theorem 9.3.3]{BCN1989} or~\cite[Section~9.6, Theorem~9.6.3]{GodsilMeagher2016}), and therefore we get the following Poincar\'e inequality:
\begin{fact}\label{fact:poincare}
  Suppose that $f\colon {\sqbinom{\mathbb{F}_2^n}{\ell}}\to\mathbb{R}$ is a function. Then
  $\inner{f}{(\mathrm{I}-\mathrm{T}_{\mathsf{Gr}(\mathbb{F}_2^n,\ell)})f}_2\geq \frac{1}{2}\mathsf{var}(f)$.
\end{fact}

\begin{claim}\label{claim:second-moment-bounds}
 Suppose  $L$ is uniformly distributed over ${\sqbinom{\mathbb{F}_2^n}{\ell}}$. Let $\mathcal{X}$ be
  a set of points of $\mathbb{F}_2^n$, and let $\mathcal{W}$ be a set of hyperplanes of $V$.
  Define $ Z_{\mathcal{X}}(L)=|\mathcal{X}\cap L|$ and  $Z_{\mathcal{W}}(L)=|\{W\in\mathcal{W}:L\subseteq W\}|$. Then 
  \begin{equation*}
    \E[Z_{\mathcal{X}}(L)^2]
    \leq
    \E[Z_{\mathcal{X}}(L)]+\E[Z_{\mathcal{X}}(L)]^2
    \qquad\text{and}\qquad
    \E[Z_{\mathcal{W}}(L)^2]
    \leq
    \E[Z_{\mathcal{W}}(L)]+\E[Z_{\mathcal{W}}(L)]^2.
  \end{equation*}
\end{claim}

\begin{proof}
  We define $  p=\mathbb{P}[x\in L]=\frac{2^\ell-1}{2^n-1}$ and $q=\mathbb{P}[L\subseteq W]=\frac{2^{n-\ell}-1}{2^n-1}$ where $x$ is any point of $\mathbb{F}_2^n$ and $W$ is any hyperplane of $\mathbb{F}_2^n$.  For distinct points $x,y$ and distinct hyperplanes $W,W'$, we have 
  \[
    \mathbb{P}[x,y\in L]
    =
    p \cdot \frac{2^\ell-2}{2^n-2}
    \leq p^2
    \qquad\text{and}\qquad
    \mathbb{P}[L\subseteq W\cap W']
    =
    q\cdot \frac{2^{n-\ell}-2}{2^n-2}
    \leq q^2.
  \]
  Consequently, it follows that 
  \begin{align*}
    \E[Z_{\mathcal{X}}(L)^2]
    &\leq |\mathcal{X}|p+|\mathcal{X}|(|\mathcal{X}|-1)p^2
    \leq \E[Z_{\mathcal{X}}(L)]+\E[Z_{\mathcal{X}}(L)]^2,
  \end{align*}
  and the result for $Z_{\mathcal{W}}(L)$ follows analogously. 
\end{proof}

\section{Proof of Main Result}
\subsection{The Fixed Dimension Stability Result}
The starting point of our proof is Theorem~\ref{thm:exact}, and we turn it into a weak form of Lemma~\ref{lem:almost_constnat_absolute} with dimension dependency.
\begin{lemma}\label{lem:weak}
    For $n,\ell$ such that $\ell,n-\ell\geq 2$, there exists $C(n,\ell)>0$ such that the following holds. Suppose that $f\colon {\sqbinom{\mathbb{F}_2^n}{\ell}}\to\{0,1\}$ satisfies that there exists $g\in J_{\leq 1}$ such that $
    \norm{f-g}_1\leq \eps$. Then either $f$ or $1-f$ is $C(n,\ell)\cdot \eps$ close to a function of the form
    \[
    0,\qquad 1_{x\in L},\qquad 1_{L\subseteq W},\qquad 1_{x\in L}+1_{L\subseteq W},
    \]
    for a point $x$ and a hyperplane $W$, and in the last case we have that $x\not\in W$.
\end{lemma}
\begin{proof}
    Let $\mathcal{F}_{n,\ell}$ be the set of all functions of the form in Theorem~\ref{thm:exact} and their complements. Then $\mathcal{F}_{n,\ell}$ is finite and its size depends only on $n,\ell$. 
    Define
    \[
    C(n,\ell) = \max_{\substack{f\colon {\sqbinom{\mathbb{F}_2^n}{\ell}}\to\{0,1\}\\ f\not\in\mathcal{F}_{n,\ell}}}\frac{{\min_{h\in \mathcal{F}_{n,\ell}}\norm{f-h}_1}}{\min_{g\in J_{\leq 1}} \norm{f-g}_1}.
    \]
    Note that by Theorem~\ref{thm:exact}, for 
    every $f\not\in\mathcal{F}_{n,\ell}$ we have that $\min_{g\in J_{\leq 1}}\norm{f-g}_1>0$, and also the number of distinct such Boolean $f$ is finite. Hence, this maximum exists and is finite, and the result follows. 
\end{proof}
We have the following immediate corollary, asserting that under the conditions of Lemma~\ref{lem:weak}, the function $f$ must have average close to $0$ or $1$. 
\begin{lemma}\label{lem:close_to_constant}
For $n,\ell$ such that $\ell,n-\ell\geq 2$, there exists $C(n,\ell)\in\mathbb{N}$ such that the following holds. Suppose that $f\colon {\sqbinom{\mathbb{F}_2^n}{\ell}}\to\{0,1\}$ satisfies that $
    \norm{f-g}_1\leq \eps$ for some $g\in J_{\leq 1}$. Then 
    \[
    \Expect{L}{f(L)}\leq C(n,\ell)\eps+2^{-\ell}+2^{-(n-\ell)}
    \qquad
    \text{or}
    \qquad
    \Expect{L}{f(L)}\geq 1-C(n,\ell)\eps-2^{-\ell}-2^{-(n-\ell)}.
    \]
\end{lemma}
\begin{proof}
    By Lemma~\ref{lem:weak} there is a function $h\colon {\sqbinom{\mathbb{F}_2^n}{\ell}}\to\{0,1\}$ as therein such that 
    $\Prob{L}{f(L)\neq h(L)}\leq C(n,\ell)\eps$. Inspecting functions in the list in Lemma~\ref{lem:weak}, we see that either $\Expect{L}{h(L)}\leq 2^{-\ell}+2^{-(n-\ell)}$ or $\Expect{L}{h(L)}\geq 1-(2^{-\ell}+2^{-(n-\ell)})$, and the result follows.
\end{proof}

\subsection{The Dimension Reduction Argument}\label{sec:dim_reduction}
Our next goal is to prove a version of Lemma~\ref{lem:close_to_constant} without the dimension dependency. Towards this end, we use the following notion of restrictions.
\begin{definition}
    Let $t,\ell,n\in \mathbb{N}$ be integers such that $t<\ell<n$, $t\leq n-\ell-2$ and $t\leq \ell-2$. A $t$-restriction consists of subspaces $(A,B,C)$ such that 
    \begin{enumerate}
        \item $A\cap B=\{0\}$ and $A+B=\mathbb{F}_2^n$.
        \item $C\subseteq B$.
        \item $\dim(A)=\ell-t$ and $\dim(C)=2t$.
    \end{enumerate}
    A $t$-random restriction is a uniformly random choice of such $(A,B,C)$.
\end{definition}
\begin{definition}
    Given a function $f\colon {\sqbinom{\mathbb{F}_2^n}{\ell}}\to\mathbb{R}$ and a $t$-restriction $(A,B,C)$, we define the function $f_{A,B,C}\colon \sqbinom{C}{t}\to\mathbb{R}$ by
    \[
    f_{A,B,C}(L) = f(A + L).
    \]
\end{definition}
We have the following basic fact.
\begin{fact}\label{fact:rest_deg_1}
    Suppose that $g\colon {\sqbinom{W}{\ell}}\to\mathbb{R}$ is a degree $1$ function, and $(A,B,C)$ is a $t$-restriction. Then $g_{A,B,C}$ is a degree $1$ function.
\end{fact}
\begin{proof}
    It suffices to prove the statement for $g(L)=1_{x\in L}$ where $x\in W$, and the general case follows by linearity. If $x\in A$, then $g_{A,B,C}\equiv 1$, and the statement is clear. Otherwise, there is a unique way to write $x=a+b$ where $a\in A$ and $b\in B\setminus \{0\}$. If $b\not \in C$ then $g_{A,B,C}\equiv 0$, and the statement is clear again. Else, 
    $g_{A,B,C}(L) = 1_{b\in L}$, which is degree $1$.
\end{proof}
For a $t$-restriction $(A,B,C)$, denote $\Delta_{A,B,C}(f) = \min(\E[f_{A,B,C}], 1-\E[f_{A,B,C}])$. The next lemma asserts that random $t$-restrictions of $f$ are close to being constant. For technical reasons, we need to formulate the assumption that $f$ is close to degree $1$ in $L_1$-norm.
\begin{lemma}\label{lem:rests_are_constants}
    Suppose that $f\colon {\sqbinom{\mathbb{F}_2^n}{\ell}}\to\{0,1\}$ satisfies
    $\norm{f-g}_1\leq \eps$ for some $g\in J_{\leq 1}$. Then for all $2\leq t\leq \ell-2, n-\ell-2$,
    \[
    \Expect{(A,B,C)}{\Delta_{A,B,C}(f)}\leq C(2t,t)\eps+2^{1-t},
    \]
    where $C(2t,t)$ is from Lemma~\ref{lem:weak}.
\end{lemma}
\begin{proof}
    For each $t$-restriction $(A,B,C)$, let $\eps_{A,B,C}(f) = \norm{f_{A,B,C}-g_{A,B,C}}_1$. Then $\Expect{(A,B,C)}{\eps_{A,B,C}(f)}=\norm{f-g}_1\leq \eps$. Applying Lemma~\ref{lem:close_to_constant} for each $A,B,C$ we get that 
    \[
    \Expect{A,B,C}{\Delta_{A,B,C}(f)}\leq \Expect{A,B,C}{C(2t,t)\eps_{A,B,C}+2^{1-t}}
    \leq C(2t,t)\eps+2^{1-t}.\qedhere
    \]
\end{proof}
In the next lemma, we show that the fact that $t$-random restrictions of $f$ are almost constant implies that $f$ itself is almost constant. 
\begin{lemma}\label{lem:almost_constant}
    Suppose $f\colon{\sqbinom{\mathbb{F}_2^n}{\ell}}\to\{0,1\}$ satisfies that $\norm{f-g}_1\leq \eps$ for some $g\in J_{\leq 1}$. Then for all $2\leq t\leq \ell-2,n-\ell-2$,
    \[
    \mathsf{var}(f)\leq 2C(2t,t)\eps + 2^{2-t}.
    \]
\end{lemma}
\begin{proof}
    Sample a random $t$-restriction $(A,B,C)$, then $L,L'\subseteq C$ of dimension $t$ that intersect at dimension $t-1$, and consider $A+L$ and $A+L'$. Then
    \[
    \Prob{(A,B,C),L,L'}{f(A+L)\neq f(A+L')}
    =\Prob{(A,B,C),L,L'}{f_{A,B,C}(L)\neq f_{A,B,C}(L')}
    \leq \Expect{(A,B,C)}{2\Delta_{A,B,C}(f)},
    \]
    which is at most $2C(2t,t)\eps+2^{2-t}$ from Lemma~\ref{lem:rests_are_constants}. On the other hand, as the distribution of $(A+L, A+L')$ is of a random edge in $\mathsf{Gr}(\mathbb{F}_2^n,\ell)$, it follows that the left hand side is equal to 
    $2\inner{f}{(\mathrm{I}-\mathrm{T}_{\mathsf{Gr}(\mathbb{F}_2^n,\ell)})f}$, which by Fact~\ref{fact:poincare} is at least $\mathsf{var}(f)$. The result follows by combining the upper bound and the lower bound.
\end{proof}
In summary, we get Lemma~\ref{lem:almost_constnat_absolute} restated below, which will be used in subsequent sections.
\MainLemma*
\begin{proof}
   Applying Lemma~\ref{lem:almost_constant} we get that 
   $\E[f]$ is either at most $4C(2t,t)\eps + 2^{3-t}$, or at least $1-(4C(2t,t)\eps + 2^{3-t})$. We choose $t = 4+\log(2/\delta)$,  
   $T = t+2$ and
   $\eps = \frac{\delta}{8 C(2t,t)}$ to get the statement.
\end{proof}

\subsection{Applying Global Hypercontractivity}
In this section, we use Lemma~\ref{lem:almost_constnat_absolute} to establish Theorem~\ref{thm:main}. To do so, we first locate the set of all points/hyperplanes on which the measure of the given function becomes noticeable, and show that they capture almost all of the mass of the function $f$. Towards this end, we make use of the following global hypercontractive inequality for the Grassmann scheme, which we present next.
\begin{definition}
    For $f\colon {\sqbinom{\mathbb{F}_2^n}{\ell}}\to\mathbb{R}$ and $x\in \mathbb{F}_2^n$, we define 
    $\mu_x(f) = \Expect{L\ni x}{f(L)}$. For a hyperplane $W\subseteq \mathbb{F}_2^n$ we define $\mu_W(f)=\Expect{L\subseteq W}{f(L)}$.
\end{definition}
\begin{definition}
    We say a function $f\colon {\sqbinom{\mathbb{F}_2^n}{\ell}}\to\{0,1\}$ is $(1,\eta)$-global if for all points $x$ and all hyperplanes $W$, 
    $\mu_x(f),\mu_W(f)\leq \eta$.
\end{definition}

The following result global hypercontractive inequality is a version of the result of~\cite{KMS2} for degree $1$ on the Grassmann scheme (see also~\cite{DKKMS2,Ellis2023analogue}). In words, the result asserts that if $h$ is a Boolean function which is $(1,\eta)$-global, then only a small amount of the $L_2$-mass lies on the first two levels.
\begin{thm}\label{thm:lvl_1_ineq}
    For all $\eta>0$, there exists $T\in\mathbb{N}$, such that the following holds for integers $\ell<n$ such that 
    $\ell, n-\ell\geq T$. If $f\colon {\sqbinom{\mathbb{F}_2^n}{\ell}}\to\{0,1\}$ is $(1,\eta)$-global, then
    \[
    \norm{P_{\leq 1} f}_2^2\leq O(\eta^{1/4} \E[f]).
    \]
\end{thm}
\begin{proof}
    The proof is by reduction to the corresponding result over the bilinear scheme, and is deferred to Section~\ref{sec:apx}.
\end{proof}

Fix $f\colon {\sqbinom{\mathbb{F}_2^n}{\ell}}\to\{0,1\}$, and without loss of generality $\E[f]\leq 1/2$ (otherwise we switch to $1-f$). We intend to use Theorem~\ref{thm:lvl_1_ineq}, and towards this end for a parameter $\eta>0$ we define
\[
\mathcal{X}_{\eta}(f) = \sett{x\in\mathbb{F}_2^n}{\mu_x(f)> \eta},
\qquad\qquad
\mathcal{W}_{\eta}(f)
=
\sett{W\subseteq\mathbb{F}_2^n\text{ hyperplane}}{\mu_W(f)> \eta}.
\]
For $L\subseteq \mathbb{F}_2^n$, let $E_{\eta}$ be the event that $x\in L$ for some $x\in \mathcal{X}_{\eta}(f)$, 
or $L\subseteq W$ for some $W\in \mathcal{W}_{\eta}(f)$. We define $h_{\eta}(L) = f(L) 1_{\overline{E_{\eta}}}(L)$. We have the following observation.

\begin{claim}\label{claim:h_global}
    The function $h$ is $(1,\eta)$-global.
\end{claim}
\begin{proof}
    For $x\in \mathcal{X}_{\eta}(f)$ we have that $\mu_x(h) = 0$. For any other $x$, we have $\mu_x(h)\leq \mu_x(f)\leq \eta$. The same argument shows that $\mu_W(h)\leq \eta$ for all hyperplanes $W$.
\end{proof}
Thus, to show the main result, we will prove that if $f$ is close to a degree $1$, then:
\begin{enumerate}
    \item For small $\eta$, the function $h_{\eta}$ is close to being $0$. Thus, 
    $f$ is close to $f1_{E_{\eta}}(L)$. This morally allows us to decompose the set of $L$ such that $f(L) = 1$ into chunks, such that the measure of $f$ in each chunk is noticeable. Here, a chunk is defined by a point $x$ or a hyperplane $W$, and corresponds to the subspaces that contain $x$ / subspaces that are contained in $W$. 
    \item We show that in almost all chunks, the measure of $f$ becomes close to $1$. Towards this end, we show that $\mathcal{X}_{\eta}(f)\setminus \mathcal{X}_{1-\eta}(f)$ is very small, which means that for most $x\in \mathcal{X}_{\eta}(f)$, $\mu_x(f)$ is close to $1$. Analogously, we show that either $\mathcal{W}_{\eta}(f)\setminus \mathcal{W}_{1-\eta}(f)$ is very small, implying that for almost all $W\in \mathcal{W}_{\eta}(f)$, $\mu_W(f)$ is close to $1$.
\end{enumerate}

The first step is established in the following claim.

\begin{claim}\label{claim:h_close_to_0}
   For all $\eps>0$, there exists $T\in\mathbb{N}$ such that the following holds for $\ell,n-\ell\geq T$. Suppose that $\norm{f-P_{\leq 1}f}_2\leq \eps \norm{f}_2$. Then
   \[
   \norm{h_{\eta}}_2
   \leq (\sqrt{\eps}+O(\eta^{1/16}))\norm{f}_2.
   \]
\end{claim}

\begin{proof}
    Set $g=P_{\leq 1} f$. Then we have that
    \[
    \inner{g}{h_{\eta}}
    =\inner{P_{\leq 1} f}{h_{\eta}}
    =\inner{f}{P_{\leq 1} h_{\eta}}
    \leq \norm{f}_2\norm{P_{\leq 1}h_{\eta}}_2,
    \]
    where the last inequality is by Cauchy--Schwarz. Using Claim~\ref{claim:h_global} and Theorem~\ref{thm:lvl_1_ineq} gives that 
    $\norm{P_{\leq 1}h_{\eta}}_2\leq O(\eta^{1/8}\norm{h_{\eta}}_2)
    \leq O(\eta^{1/8}\norm{f}_2)$, and plugging this above gives
    \begin{equation}\label{eq:1}
    \inner{g}{h_{\eta}}\leq 
    O\left(\eta^{1/8}\norm{f}_2^2\right).
    \end{equation}
    Next, note that
    \begin{equation}\label{eq:2}
    \inner{g}{h_{\eta}}
    = 
    \inner{f}{h_{\eta}}
    -
    \inner{f-g}{h_{\eta}}
    \geq 
    \inner{f}{h_{\eta}}
    -\norm{f-g}_2\norm{h_{\eta}}_2
    \geq 
    \inner{f}{h_{\eta}}
    -\eps \norm{f}_2^2,
    \end{equation}
    where in the second transition we used Cauchy--Schwarz and in the last transition we used 
    $\norm{f-g}_2\leq \eps \norm{f}_2$ and $\norm{h_{\eta}}\leq \norm{f}_2$. Combining~\eqref{eq:1} and~\eqref{eq:2} gives that 
    $\inner{f}{h_{\eta}}\leq (\eps+O(\eta^{1/8}))\norm{f}_2^2$. Noting that $\inner{f}{h_{\eta}}=\norm{h_{\eta}}_2^2$ finishes the proof.
\end{proof}

We now establish the second step.
\begin{claim}\label{claim:get_zoom_in_to_1}
    For all $\delta,\xi\in(0,1)$, there exist $T\in\mathbb{N}$ and $\eps>0$ such that the following holds for $\ell,n-\ell\geq T$.
    Suppose $\norm{f-P_{\leq 1}f}_2\leq \eps\norm{f}_2$ and that $\mathcal{X}\subseteq \mathcal{X}_{\delta}(f)$ has size $\xi\E[f] \frac{2^{n}-1}{2^{\ell}-1}$. Then there exists $x\in \mathcal{X}$ such that
    \[
    \mu_x(f)\geq 1-\delta.
    \]
\end{claim}
\begin{proof}
    Fix $\delta,\xi>0$ and take $T,\eps$ from Lemma~\ref{lem:almost_constnat_absolute}. We pick $\eps' = \eps\sqrt{\xi}/10$ in the claim.
    Let $g = P_{\leq 1}f$. 
    For a subspace $L\subseteq\mathbb{F}_2^n$ define $Z[L] = \card{L\cap \mathcal{X}}$. 
    Then $\Expect{L}{Z[L]}=\xi\E[f]$ and 
    $\Expect{L}{Z[L]^2}\leq
    \E[Z]+\E[Z]^2\leq 10\xi \E[f]$.
    Letting $N = \card{{\sqbinom{\mathbb{F}_2^n}{\ell}}}$, we get
    \begin{align*}
    \Expect{x\in \mathcal{X}}{\Expect{L\ni x}{\card{f(L)-g(L)}}}
    &\leq 
    \sum\limits_{L}\frac{Z(L)}{N\xi\E[f]}\card{f(L)-g(L)}\\
    &\leq 
    \frac{1}{\xi\E[f]}\sqrt{\sum\limits_{L}\frac{1}{N}\card{f(L)-g(L)}^2}
    \sqrt{\sum\limits_{L}\frac{Z(L)^2}{N}}\\
    &=\frac{1}{\xi\E[f]}\norm{f-g}_2\sqrt{\Expect{L}{Z[L]^2}}\\
    &\leq \frac{1}{\xi\E[f]}
    \eps'\norm{f}_2 \sqrt{10\xi\E[f]},
    \end{align*}
    which is at most $\eps$.
    It follows that there exists $x\in \mathcal{X}$ such that 
    $\Expect{L\ni x}{\card{f(L)-g(L)}}\leq \eps$. Fixing a hyperplane $B$ such that $\mathsf{span}(x)+B = \mathbb{F}_2^n$, and defining 
    $f_x(L') = f(\mathsf{span}(x)+L')$, 
    $g_x(L') = g(\mathsf{span}(x)+L')$, we have
    \[
    \norm{f_x-g_x}_1=\Expect{L'\in \sqbinom{B}{\ell-1}}{\card{f_x(L')-g_x(L')}}\leq \eps.
    \]
    As $g\in J_{\leq 1}$, by the same argument as in Fact~\ref{fact:rest_deg_1} we have that $g_x\in J_{\leq 1}$. It follows from Lemma~\ref{lem:almost_constnat_absolute} that either $\E[f_{x}]\geq 1-\delta$ or $\E[f_x]\leq \delta$, and as $\E[f_x] = \mu_x(f)> \delta$, we conclude it must be the former case.
\end{proof}
We have the following analogous statement for $\mathcal{W}_{\delta}(f)$.
\begin{claim}\label{claim:get_zoom_ot_to_1}
    For all $\delta,\xi \in(0,1)$, there exist $T\in\mathbb{N}$ and $\eps>0$ such that the following holds for $\ell,n-\ell\geq T$.
    Suppose $\norm{f-P_{\leq 1}f}_2\leq \eps\norm{f}_2$ and that $\mathcal{W}\subseteq \mathcal{W}_{\delta}(f)$ has size $\xi\E[f] \frac{2^{n}-1}{2^{n-\ell}-1}$. Then there exists $W\in \mathcal{W}$ such that
    \[
    \mu_W(f)\geq 1-\delta.
    \]
\end{claim}
\begin{proof}
    Fix $\delta,\xi>0$ and take $T,\eps$ from Lemma~\ref{lem:almost_constnat_absolute}. We pick $\eps' = \frac{\eps\sqrt{\xi}}{10}$.
    Let $g = P_{\leq 1}f$. 
    For a subspace $L\subseteq\mathbb{F}_2^n$ define $Z[L] = \#\sett{W\in \mathcal{W}}{L\subseteq W}$. 
    Then $\Expect{L}{Z[L]}=\xi\E[f]$ and 
    $\Expect{L}{Z[L]^2}\leq 10\xi \E[f]$.
    By the same computation as in Claim~\ref{claim:get_zoom_in_to_1} we get that $\Expect{W\in \mathcal{W}}{\Expect{L\subseteq W}{\card{f(L)-g(L)}}}\leq \eps$. It follows that there exists $W\in \mathcal{W}$ such that 
    $\Expect{L\subseteq W}{\card{f(L)-g(L)}}\leq \eps$. Defining, $f_W, g_W\colon {\sqbinom{W}{\ell}}\to\mathbb{R}$ by
    $f_W(L') = f(L')$, 
    $g_W(L') = g(L')$, we have
    \[
    \norm{f_W-g_W}_1=\Expect{L'\in {\sqbinom{W}{\ell}}}{\card{f_W(L')-g_W(L')}}\leq \eps.
    \]
    As $g\in J_{\leq 1}$, by the same argument as in Fact~\ref{fact:rest_deg_1} we have that $g_W\in J_{\leq 1}$. It follows from Lemma~\ref{lem:almost_constnat_absolute} that either $\E[f_{W}]\geq 1-\delta$ or $\E[f_W]\leq \delta$, and as $\E[f_W] = \mu_W(f)> \delta$, we conclude it must be the former case.
\end{proof}

We use the above two claims to show that the sets 
$\mathcal{X}_{\delta}(f)\setminus \mathcal{X}_{1-2\delta}(f)$, $\mathcal{W}_{\delta}(f)\setminus \mathcal{W}_{1-2\delta}(f)$ must have size which is small, relative to the measure of $f$.
\begin{lemma}\label{lem:most_are_close_to_1}
  For all $\delta,\xi\in(0,1)$, there exist $T$ and $\eps$ such that the following holds for $\ell,n-\ell\geq T$. Suppose  
  $\norm{f-P_{\leq 1}f}_2\leq \eps\norm{f}_2$, then
  \[
  \card{\mathcal{X}_{\delta}(f)\setminus \mathcal{X}_{1-2\delta}(f)}<\xi\E[f] 2^{n-\ell}
  \qquad
  \card{\mathcal{W}_{\delta}(f)\setminus \mathcal{W}_{1-2\delta}(f)}<\xi\E[f] 2^{\ell}.
  \]
\end{lemma}
\begin{proof}
    Fix $\delta,\xi>0$, pick $T_1,\eps_1$ from Claim~\ref{claim:get_zoom_in_to_1}, $T_2,\eps_2$ from Claim~\ref{claim:get_zoom_ot_to_1} for $\delta$ and $\xi/2$, and choose $T = \max(T_1,T_2)$, $\eps = \min(\eps_1,\eps_2)$. Note that if 
    $\card{\mathcal{X}_{\delta}(f)\setminus \mathcal{X}_{1-2\delta}(f)}\geq \xi\E[f] 2^{n-\ell}$, then we could find 
    $\mathcal{X}\subseteq \mathcal{X}_{\delta}(f)$ of size at $\xi\E[f]2^{n-\ell}$ such that $\mu_x(f)\leq 1-2\delta$ for all $x\in\mathcal{X}$, but this contradicts Claim~\ref{claim:get_zoom_in_to_1}. This establishes the first assertion, and the second one follows similarly.
\end{proof}

For technical reasons, we also need a simple upper bound on the size of the sets $\mathcal{X}_{1-2\delta}$, 
$\mathcal{W}_{1-2\delta}$ given in the following lemma.
\begin{lemma}\label{lem:trivial_upper_bound_on_sizes}
    For $f\colon {\sqbinom{\mathbb{F}_2^n}{\ell}}\to\{0,1\}$ with $\E[f]\leq 1/2$ and $\delta\leq 1/8$ we have that 
    \[
    \card{\mathcal{X}_{1-2\delta}(f)}\leq O(\E[f] 2^{n-\ell}),
    \qquad 
    \card{\mathcal{W}_{1-2\delta}(f)}\leq O(\E[f] 2^{\ell}).
    \]
\end{lemma}
\begin{proof}
We define $Z(L)=|\mathcal{X}\cap L|$, and note that
$\E[Z] = \frac{2^\ell-1}{2^n-1} |\mathcal X|$ and
\[ \E[fZ] =\sum_{x\in\mathcal{X}}\E[f(L)1_{x\in L}] =\mathbb{P}[x\in L]\sum_{x\in\mathcal{X}}\mu_x(f)\geq (1-2\delta)\mathbb{P}[x\in L]|\mathcal{X}| =(1-2\delta) \E[Z]. \]

On the other hand, we may apply Claim~\ref{claim:second-moment-bounds} to get
\[\E[fZ]^2 \leq  \E[f^2]\E[Z^2]\leq \E[f] (\E[Z]+\E[Z]^2),\]
where the first inequality follows from the Cauchy--Schwarz inequality and the second implicitly uses the Booleanity of $f$ so that $f^2 = f$. 

Combining the upper and lower bounds above gives $\bigl((1-2\delta)^2-\E[f]\bigr)\E[Z]\le \E[f]$. Since $\delta\le 1/8$ and $\E[f]\le 1/2$, it follows that $(1-2\delta)^2-\E[f]\ge\tfrac1{16}$, which shows that $|\mathcal{X}|\le 16\E[f]\frac{2^n-1}{2^\ell-1} = O(\E[f] 2^{n-\ell})$. The hyperplane statement follows from duality.
\end{proof}

\subsection{Proof of Theorem~\ref{thm:main}}
We assume without loss of generality that $\E[f]\leq 1/2$ (otherwise we work with $1-f$).
Fix $\delta>0$, let $c>0$ be a sufficiently small absolute constant, and take $\delta' = \xi = c \delta^{8}$. Pick $T_1$ and $\eps_1$ from Lemma~\ref{lem:most_are_close_to_1}, $T_2$, $\eps_2$ from Theorem~\ref{thm:lvl_1_ineq} for $\delta'$, $T_3$, $\eps_3$ from Lemma~\ref{lem:almost_constnat_absolute} for $\delta'$, $T_4$ from Claim~\ref{claim:h_global} for $\delta'$, then set $T = \max(T_1,T_2,T_3,T_4)$ and $\eps = \min(\eps_1,\eps_2,\eps_3,\delta')$.
We prove the statement for 
$\mathcal{X} = \mathcal{X}_{1-2\delta'}(f)$ 
and $\mathcal{W} = \mathcal{W}_{1-2\delta'}(f)$. We note that by Lemma~\ref{lem:almost_constnat_absolute} we already know that $\E[f]\leq \delta'$, so 
$\mathsf{var}(f)\leq \delta$.

Define 
\[
f_1(L) = f(L) 1_{E_{\delta'}}(L),
\qquad
f_2(L) = f(L) 1_{E_{1-2\delta'}}(L),
\qquad g(L) = \sum\limits_{x\in\mathcal{X}} 1_{x\in L}+\sum\limits_{W\in\mathcal{W}} 1_{L\subseteq W}.
\]
By Claim~\ref{claim:h_close_to_0} we get that
\begin{equation}\label{eq:3}
  \norm{f-f_1}_2^2 = \norm{h_{\delta'}}_2^2
  \leq O(\delta'^{1/8}\E[f]).
\end{equation}
Next, note that $1_{E_{1-2\delta'}}\leq 1_{E_{\delta'}}$ pointwise and by Lemma~\ref{lem:most_are_close_to_1}
\[
\Prob{L}{L\cap (\mathcal{X}_{\delta'}(f)
\setminus 
\mathcal{X}_{1-2\delta'}(f))\neq \emptyset}
\leq 
\card{\mathcal{X}_{\delta'}(f)
\setminus 
\mathcal{X}_{1-2\delta'}(f)}2^{\ell-n}\leq \xi \E[f],
\]
\[
\Prob{L}{\exists W\in\mathcal{W}_{\delta'}(f)
\setminus 
\mathcal{W}_{1-2\delta'}(f),  L\subseteq W}
\leq 
\card{\mathcal{W}_{\delta'}(f)
\setminus 
\mathcal{W}_{1-2\delta'}(f)}2^{-\ell}\leq \xi \E[f],
\]
implying that $\norm{1_{E_{\delta'}}-1_{E_{1-2\delta'}}}_1\leq 2\xi\E[f]$. Thus, we get that
\begin{equation}\label{eq:4}
  \norm{f_1-f_2}_2^2 
  \leq \norm{1_{E_{\delta'}}-1_{E_{1-2\delta'}}}_1
  \leq 2\xi\E[f].
\end{equation}

We define $ G(L)=1_{g(L)\ge 1}$ to be the Boolean-valued rounding of $g$. Now, we observe that $E_{1-2\delta'}(L)$ occurs if and only if $g(L)\ge 1$ and so $f_2 = fG$. So to estimate the norm of $g-f_2$, we use  the triangle inequality to get $\norm{g-f_2}_2^2
  \le 2\norm{g-G}_2^2+2\norm{G-f_2}_2^2$, and we next bound each of the terms in turn. 
  
  For the first term, we have $(g-G)^2=(g-1)^21_{g\ge 2}
  \le g(g-1)$. 
  As in Claim~\ref{claim:second-moment-bounds}, we define $p = \mathbb P[x \in L] = \frac{2^\ell-1}{2^n-1}$ and $q = \mathbb P[L \subseteq W] = \frac{2^{n-\ell}-1}{2^n-1}$, where $x$ is any point and $W$ is any hyperplane. As observed there, for distinct points $x,y$ and distinct hyperplanes $W,W'$ we have $\mathbb P[x,y \in L] \le p^2$ and $\mathbb P[L \subseteq W \cap W'] \le q^2$. If $x$ is a point and $W$ is a hyperplane then either $x \notin W$, in which case $\mathbb P[x \in L \subseteq W] = 0$, or $x \in W$, in which case
  \[
   \mathbb P[x \in L \subseteq W] = \frac{(2^\ell-1)(2^{n-\ell}-1)}{(2^n-1)(2^{n-1}-1)} = \frac{2^n-1}{2^{n-1}-1} pq \le 3pq.
  \]
  Therefore
\begin{align*}
  \E[g(g-1)]
  &=\sum_{\substack{x,y\in\mathcal{X}\\x\ne y}}\mathbb P[x,y\in L]
    +\sum_{\substack{W,W'\in\mathcal{W}\\W\ne W'}}\mathbb P[L\subseteq W\cap W']+2\sum_{\substack{x\in\mathcal{X}\\W\in\mathcal{W}}}\mathbb P[x\in L\subseteq W]\\
  &\le |\mathcal{X}|(|\mathcal{X}|-1)p^2
    +|\mathcal{W}|(|\mathcal{W}|-1)q^2
    +6|\mathcal{X}||\mathcal{W}|pq\\
  &\le 3\bigl(|\mathcal{X}|p+|\mathcal{W}|q\bigr)^2 \leq O(\E[f]^2),
\end{align*}
where the final inequality follows from Lemma~\ref{lem:trivial_upper_bound_on_sizes}. Consequently, we have $\norm{g-G}_2^2=O\bigl(\E[f]^2\bigr)$.

For the second term, since $f$ and $G$ are Boolean and $f_2=fG$,
\begin{align*}
  \norm{G-f_2}_2^2 =\E[(1-f)G] &\le \sum_{x\in\mathcal{X}}\E[(1-f)1_{x\in L}]
    +\sum_{W\in\mathcal{W}}\E[(1-f)1_{L\subseteq W}]\\
  &=p\sum_{x\in\mathcal{X}}(1-\mu_x(f))
    +q\sum_{W\in\mathcal{W}}(1-\mu_W(f))\\
  &\le 2\delta'\bigl(|\mathcal{X}|p+|\mathcal{W}|q\bigr)\leq O\bigl(\delta'\E[f]\bigr),
\end{align*}
where the last inequality again uses Lemma~\ref{lem:trivial_upper_bound_on_sizes}. Combining the above and using $\E[f]\leq \delta'$, we conclude
\begin{equation}\label{eq:5}
  \norm{g-f_2}_2^2 
  \leq O(\delta' \E[f]).
\end{equation}
  
Combining~\eqref{eq:3},~\eqref{eq:4},~\eqref{eq:5} we get that
\[
\norm{f-g}_2^2\leq 
O((\delta'^{1/8}+\xi+\delta')\E[f])\leq \delta\E[f],
\]
where the last transition holds provided that $c$ is small enough.
\section{Discussion and Open Problems}
Our main result, Theorem~\ref{thm:main}, gives an analog of the FKN theorem for the Grassmann scheme so long as the dimensions $\ell$ and $n-\ell$ are sufficiently large. This feature is crucially used in our key intermediate result, Lemma~\ref{lem:almost_constnat_absolute}, which morally speaking asserts that at a coarse scale, degree $1$ functions in this case must be close to constant. We believe that this phenomenon should occur for all constant degree functions on the Grassmann scheme, which will provide an analog of the Kindler--Safra theorem for the Grassmann scheme. 

The arguments given in this paper seem to generalize well to all constant degrees $d$, and the missing ingredient is the analog of the result of~\cite{FilmusIhringer} for degree $d$ Boolean functions over the Grassmann scheme. Recent works~\cite[Conjecture 7.2]{deg2Gras},~\cite{filmus2019boolean} proposed candidate conjectured structure for such functions, suggesting a structure analogous to constant depth decision trees. However, a quick exploration using ChatGPT 5.6 Pro yields the following counterexample. For $i=1,\ldots,n$ pick points $x(1),\ldots,x(n)$ and hyperplanes $W(1),\ldots,W(n)$ such that $x(i)\in W(j)$ if and only if $i=j$. A concrete choice can be taken to be
\[
x(i) = e_i\in\mathbb{F}_2^n,
\qquad
W(i) = \left\{x\in \mathbb{F}_2^n \;\middle|\; \sum\limits_{j\neq i}x_j=0\right\}.
\]
Then the function 
$f(L) = \sum\limits_{i=1}^{n} 1_{x(i)\in L\subseteq W(i)}$ is a degree $2$ Boolean function, however it is not of the structure asserted by either of the aforementioned conjectures. For the sake of the argument presented herein, it suffices to prove a more modest statement, namely a version of Lemma~\ref{lem:almost_constnat_absolute} where $g\in J_{\leq d}$ (instead of $g\in J_{\leq 1}$). We conjecture that such a result is true, though we do not currently know how to establish it.
\bibliography{ref}

@article{KMS2,
  author  = {Khot, Subhash and Minzer, Dor and Safra, Muli},
  title   = {Pseudorandom sets in {Grassmann} graphs have near-perfect expansion},
  journal = {Annals of Mathematics},
  volume  = {198},
  number  = {1},
  pages   = {1--92},
  year    = {2023},
  doi     = {10.4007/annals.2023.198.1.1}
}

@article{FilmusIhringer,
  author  = {Filmus, Yuval and Ihringer, Ferdinand},
  title   = {Boolean degree 1 functions on some classical association schemes},
  journal = {Journal of Combinatorial Theory, Series A},
  volume  = {162},
  pages   = {241--270},
  year    = {2019},
  doi     = {10.1016/j.jcta.2018.11.006}
}

@article{DinurFriedgutKindlerODonnell2007,
  author  = {Dinur, Irit and Friedgut, Ehud and Kindler, Guy and
             O'Donnell, Ryan},
  title   = {On the {Fourier} tails of bounded functions over the discrete cube},
  journal = {Israel Journal of Mathematics},
  volume  = {160},
  pages   = {389--412},
  year    = {2007},
  doi     = {10.1007/s11856-007-0068-9}
}

@unpublished{KS,
  author = {Kindler, Guy and Safra, Shmuel},
  title  = {Noise-resistant {Boolean} functions are juntas},
  year   = {2004},
  month  = mar,
  note   = {Unpublished manuscript},
  url    = {https://www.cs.tau.ac.il/~safra/PapersAndTalks/juntas.pdf}
}

@article{Bourgain2002FourierSpectrum,
  author  = {Bourgain, Jean},
  title   = {On the distribution of the {Fourier} spectrum of {Boolean} functions},
  journal = {Israel Journal of Mathematics},
  volume  = {131},
  pages   = {269--276},
  year    = {2002},
  doi     = {10.1007/BF02785861}
}

@article{Filmus2020FKNMultislice,
  author        = {Filmus, Yuval},
  title         = {{FKN} theorem for the multislice, with applications},
  journal       = {Combinatorics, Probability and Computing},
  volume        = {29},
  number        = {2},
  pages         = {200--212},
  year          = {2020},
  doi           = {10.1017/S0963548319000361},
  eprint        = {1809.03089},
  archivePrefix = {arXiv},
  primaryClass  = {math.CO}
}

@article{EldanKindlerLifshitzMinzer2025,
  author        = {Eldan, Ronen and Kindler, Guy and Lifshitz, Noam and
                   Minzer, Dor},
  title         = {Isoperimetric inequalities made simpler},
  journal       = {Discrete Analysis},
  volume        = {7},
  pages         = {1--23},
  year          = {2025},
  doi           = {10.19086/da.142095},
  eprint        = {2204.06686},
  archivePrefix = {arXiv},
  primaryClass  = {math.CO},
  url           = {https://doi.org/10.19086/da.142095}
}

@article{AlonDinurFriedgutSudakov2004,
  author  = {Alon, Noga and Dinur, Irit and Friedgut, Ehud and
             Sudakov, Benny},
  title   = {Graph products, {Fourier} analysis and spectral techniques},
  journal = {Geometric and Functional Analysis},
  volume  = {14},
  number  = {5},
  pages   = {913--940},
  year    = {2004},
  doi     = {10.1007/s00039-004-0475-2}
}

@article{JendrejOleszkiewiczWojtaszczyk2012,
  author  = {Jendrej, Jacek and Oleszkiewicz, Krzysztof and
             Wojtaszczyk, Jakub O.},
  title   = {On some extensions of the {Friedgut--Kalai--Naor} theorem},
  journal = {Theory of Computing},
  volume  = {8},
  number  = {1},
  pages   = {261--281},
  year    = {2012},
  doi     = {10.4086/toc.2012.v008a012}
}

@misc{RubinsteinSafra2015,
      title={Boolean functions whose {F}ourier transform is concentrated on pair
wise disjoint subsets of the input}, 
      author={Aviad Rubinstein and Muli Safra},
      year={2015},
      eprint={1512.09045},
      archivePrefix={arXiv},
      primaryClass={math.PR},
      url={https://arxiv.org/abs/1512.09045}
}

@article{Ihringer2023Classification,
    AUTHOR = {Ihringer, Ferdinand},
     TITLE = {The classification of {B}oolean degree 1 functions in
              high-dimensional finite vector spaces},
   JOURNAL = {Proc. Amer. Math. Soc.},
  FJOURNAL = {Proceedings of the American Mathematical Society},
    VOLUME = {152},
      YEAR = {2024},
    NUMBER = {12},
     PAGES = {5355--5365},
      ISSN = {0002-9939,1088-6826},
   MRCLASS = {51E20 (05E30 06E30)},
  MRNUMBER = {4855888},
MRREVIEWER = {Valentino\ Smaldore},
       DOI = {10.1090/proc/16957},
       URL = {https://doi.org/10.1090/proc/16957},
}

@misc{Filmus2026Grassmann,
  author        = {Filmus, Yuval},
  title         = {Boolean degree one functions on the {Grassmann} scheme},
  year          = {2026},
  eprint        = {2606.23465},
  archivePrefix = {arXiv},
  primaryClass  = {math.CO}
}

@article{KMS1,
  author  = {Khot, Subhash and Minzer, Dor and Safra, Muli},
  title   = {On independent sets, 2-to-2 games, and {Grassmann} graphs},
  journal = {Theory of Computing},
  volume  = {21},
  number  = {10},
  pages   = {1--55},
  year    = {2025},
  doi     = {10.4086/toc.2025.v021a010}
}

@article{DKKMS1,
  author  = {Dinur, Irit and Khot, Subhash and Kindler, Guy and
             Minzer, Dor and Safra, Muli},
  title   = {Towards a proof of the 2-to-1 games conjecture},
  journal = {Theory of Computing},
  volume  = {21},
  number  = {11},
  pages   = {1--50},
  year    = {2025},
  doi     = {10.4086/toc.2025.v021a011}
}

@article{DKKMS2,
  author  = {Dinur, Irit and Khot, Subhash and Kindler, Guy and
             Minzer, Dor and Safra, Muli},
  title   = {On non-optimally expanding sets in {Grassmann} graphs},
  journal = {Israel Journal of Mathematics},
  volume  = {243},
  number  = {1},
  pages   = {377--420},
  year    = {2021},
  doi     = {10.1007/s11856-021-2164-7}
}

@article{barak2015making,
  author  = {Barak, Boaz and Gopalan, Parikshit and H{\aa}stad, Johan and
             Meka, Raghu and Raghavendra, Prasad and Steurer, David},
  title   = {Making the long code shorter},
  journal = {SIAM Journal on Computing},
  volume  = {44},
  number  = {5},
  pages   = {1287--1324},
  year    = {2015}
}

@article{khot2017hardness,
  author  = {Khot, Subhash and Saket, Rishi},
  title   = {Hardness of coloring 2-colorable 12-uniform hypergraphs with
             {$2^{(\log n)^{\Omega(1)}}$} colors},
  journal = {SIAM Journal on Computing},
  volume  = {46},
  number  = {1},
  pages   = {235--271},
  year    = {2017}
}

@article{KaufmanM,
  author  = {Kaufman, Tali and Minzer, Dor},
  title   = {Improved optimal testing results from global hypercontractivity},
  journal = {SIAM Journal on Computing},
  volume  = {54},
  number  = {3},
  pages   = {625--663},
  year    = {2025}
}

@inproceedings{MZheng,
  author    = {Minzer, Dor and Zheng, Kai Zhe},
  title     = {Optimal testing of generalized {Reed--Muller} codes in fewer
               queries},
  booktitle = {Proceedings of the 64th Annual IEEE Symposium on Foundations
               of Computer Science (FOCS)},
  pages     = {206--233},
  year      = {2023},
  publisher = {IEEE}
}

@book{BCN1989,
  author    = {Brouwer, Andries E. and Cohen, Arjeh M. and Neumaier, Arnold},
  title     = {Distance-Regular Graphs},
  series    = {Ergebnisse der Mathematik und ihrer Grenzgebiete. 3. Folge},
  volume    = {18},
  publisher = {Springer-Verlag},
  address   = {Berlin},
  year      = {1989},
  doi       = {10.1007/978-3-642-74341-2},
  isbn      = {978-3-540-50619-5}
}

@book{GodsilMeagher2016,
  author    = {Godsil, Chris and Meagher, Karen},
  title     = {Erd{\H{o}}s--Ko--Rado Theorems: Algebraic Approaches},
  series    = {Cambridge Studies in Advanced Mathematics},
  volume    = {149},
  publisher = {Cambridge University Press},
  address   = {Cambridge},
  year      = {2016},
  doi       = {10.1017/CBO9781316414958},
  isbn      = {978-1-107-12844-6}
}

@article{friedgut2002boolean,
  author  = {Friedgut, Ehud and Kalai, Gil and Naor, Assaf},
  title   = {Boolean functions whose {Fourier} transform is concentrated on
             the first two levels},
  journal = {Advances in Applied Mathematics},
  volume  = {29},
  number  = {3},
  pages   = {427--437},
  year    = {2002}
}

@article{deg2Gras,
  author   = {De Beule, Jan and D'haeseleer, Jozefien and Ihringer, Ferdinand and
              Mannaert, Jonathan},
  title    = {Degree 2 {Boolean} functions on {Grassmann} graphs},
  journal  = {Electronic Journal of Combinatorics},
  volume   = {30},
  number   = {1},
  pages    = {1--23},
  year     = {2023},
  doi      = {10.37236/11040},
  issn     = {1097-1440}
}

@article{filmus2019boolean,
  author  = {Filmus, Yuval and Ihringer, Ferdinand},
  title   = {Boolean constant degree functions on the slice are juntas},
  journal = {Discrete Mathematics},
  volume  = {342},
  number  = {12},
  pages   = {111614},
  year    = {2019}
}

@article{Filmus2021BooleanSn,
  author        = {Filmus, Yuval},
  title         = {Boolean functions on {$S_n$} which are nearly linear},
  journal       = {Discrete Analysis},
  volume        = {25},
  pages         = {1--27},
  year          = {2021},
  doi           = {10.19086/da.30186},
  eprint        = {2107.07833},
  archivePrefix = {arXiv},
  primaryClass  = {math.CO}
}

@article {EFF1,
    AUTHOR = {Ellis, David and Filmus, Yuval and Friedgut, Ehud},
     TITLE = {A quasi-stability result for dictatorships in {$S_n$}},
   JOURNAL = {Combinatorica},
  FJOURNAL = {Combinatorica. An International Journal on Combinatorics and
              the Theory of Computing},
    VOLUME = {35},
      YEAR = {2015},
    NUMBER = {5},
     PAGES = {573--618},
      ISSN = {0209-9683},
   MRCLASS = {05D99 (20B30 42B10)},
  MRNUMBER = {3437896},
MRREVIEWER = {Norihide Tokushige},
       DOI = {10.1007/s00493-014-3027-1},
}

@article {EFF2,
    AUTHOR = {Ellis, David and Filmus, Yuval and Friedgut, Ehud},
     TITLE = {A stability result for balanced dictatorships in {${\rm
              S_n}$}},
   JOURNAL = {Random Structures Algorithms},
  FJOURNAL = {Random Structures \& Algorithms},
    VOLUME = {46},
      YEAR = {2015},
    NUMBER = {3},
     PAGES = {494--530},
      ISSN = {1042-9832},
   MRCLASS = {94D05 (06E30)},
  MRNUMBER = {3324758},
MRREVIEWER = {Norihide Tokushige},
       DOI = {10.1002/rsa.20515},
}

@article {EFF3,
    AUTHOR = {Ellis, David and Filmus, Yuval and Friedgut, Ehud},
     TITLE = {Low-degree {B}oolean functions on {$S_n$}, with an application
              to isoperimetry},
   JOURNAL = {Forum Math. Sigma},
  FJOURNAL = {Forum of Mathematics. Sigma},
    VOLUME = {5},
      YEAR = {2017},
     PAGES = {Paper No. e23, 46},
   MRCLASS = {05D05 (05C25 06E30 94D05)},
  MRNUMBER = {3708207},
MRREVIEWER = {Mikl\'{o}s B\'{o}na},
       DOI = {10.1017/fms.2017.24},
}

@inproceedings{Ellis2023analogue,
author = {Ellis, David and Kindler, Guy and Lifshitz, Noam},
title = {An Analogue of Bonami’s Lemma for Functions on Spaces of Linear Maps, and 2-2 Games},
year = {2023},
isbn = {9781450399135},
publisher = {Association for Computing Machinery},
address = {New York, NY, USA},
url = {https://doi.org/10.1145/3564246.3585116},
doi = {10.1145/3564246.3585116},
booktitle = {Proceedings of the 55th Annual ACM Symposium on Theory of Computing},
pages = {656–660},
numpages = {5},
location = {Orlando, FL, USA},
series = {STOC 2023}
}
\bibliographystyle{alphaurl}
\appendix
\section{Missing Proofs}\label{sec:apx}
In this section we prove Theorem~\ref{thm:lvl_1_ineq} by reducing it to the analogous result over the bilinear scheme.
The bilinear scheme $H_{n,\ell}$ has the vertex set $\mathbb{F}_2^{n\times \ell}$, and its edges are $(A,B)$ where $B-A = u\otimes v$ for non-zero $u\in\mathbb{F}_2^n$ and $v\in\mathbb{F}_2^{\ell}$. Thus, one can naturally lift a function $f\colon {\sqbinom{\mathbb{F}_2^n}{\ell}}\to\{0,1\}$ to a function $f^{\star}$ on $H_{n,\ell}$ by defining
\[
f^{\star}(M) = f(\mathsf{span}(\mathsf{col}_1(M),\ldots,\mathsf{col}_{\ell}(M))
\]
if $M$ has rank $\ell$, and $0$ otherwise. We note that $f^{\star}$ is basis invariant, in the sense that $f^{\star}(MB) = f^{\star}(M)$ for all invertible $B\in\mathbb{F}_2^{\ell\times\ell}$.
Define $c = \prod\limits_{i=0}^{\ell-1}(1-2^{i-n})$, and note that $c$ is exactly the probability that the rank of $M\in\mathbb{F}_2^{n\times \ell}$  has rank $\ell$. We will use the fact that $c\geq \Omega(1)$, which is clear. The bilinear scheme is a Cayley graph, and as such its eigenvectors are characters. It is shown in~\cite{KMS2} that the level decomposition is given as
\[
P_{=i} (f^{\star})(M)
=\sum\limits_{\alpha\in\mathbb{F}_2^{n\times \ell}, \mathsf{rk}(\alpha)=i}\widehat{f^{\star}}(\alpha)\chi_{\alpha}(M),
\]
where $\chi_{\alpha}(M) = (-1)^{\sum_{i,j}\alpha_{i,j}M_{i,j}}$, $\widehat{f^{\star}}(\alpha) = \Expect{M}{f^{\star}(M)\chi_{\alpha}(M)}$.

We will use the definition of pseudorandomness of Boolean functions over $H_{n,\ell}$, as per~\cite[Definition 2.5]{KMS2}. Below, we specialize it to our case.
\begin{definition}
    For a basis invariant function $f^{\star}\colon \mathbb{F}_2^{n\times \ell}\to\{0,1\}$, we say that it is $(1,\eta)$-pseudorandom if for all points $x\in\mathbb{F}_2^n$ and hyperplanes 
    $W\subseteq \mathbb{F}_2^n$
    \[
    \Prob{x_2,\ldots,x_{\ell}}{f^{\star}(x,x_2,\ldots,x_{\ell})=1}\leq \eta,
    \qquad
    \Prob{x_1,x_2,\ldots,x_{\ell}\in W}{f^{\star}(x_1,x_2,\ldots,x_{\ell})=1}\leq \eta.
    \]
\end{definition}
We have the following statement, which is~\cite[Lemma 2.8]{KMS2}
\begin{fact}\label{fact:global}
    If $f\colon {\sqbinom{\mathbb{F}_2^n}{\ell}}\to\{0,1\}$ is $(1,\eta)$-global, then 
    $f^{\star}$ is 
    $(1,\eta+2^{\ell+1-n})$-pseudorandom.
\end{fact}

We will also use the following result, which is~\cite[Theorem 2.13]{KMS2}.
\begin{lemma}\label{lem:lvl_1_bilinear}
    Suppose that a basis invariant $f^{\star}\colon \mathbb{F}_2^{n\times \ell}\to\{0,1\}$ is $(1,\eta)$-pseudorandom. Then
    \[
    \norm{P_{=1}(f^{\star})}_{2}^2\leq 2^{10}\eta^{1/4} \E[f^{\star}].
    \]
\end{lemma}

\subsection{Proof of 
Theorem~\ref{thm:lvl_1_ineq}}\label{sec:apx_1}
   Fix $f$ as in the theorem and let $f^{\star}$ be its lifting to $H_{n,\ell}$. We prove that
   \begin{equation}\label{eq:8}
   \norm{P_{=1}f}_2^2
   \leq
   O(\norm{P_{=1}(f^{\star})}^2+2^{\ell-n}\E[f]^2).
   \end{equation}
   Plugging this into $\norm{P_{\leq 1}f}_2^2=\E[f]^2+\norm{P_{=1}f}_2^2$, using Lemma~\ref{lem:lvl_1_bilinear} and $\E[f]\leq \eta$, the result follows. 
   
   Towards establishing~\eqref{eq:8}, define the operator $\mathrm{S}\colon L_2({\sqbinom{\mathbb{F}_2^n}{\ell}})
   \to 
   L_2(\sqbinom{\mathbb{F}_2^n}{n-1})$ by
   \[
   \mathrm{S} h(W) =
   \Expect{L\in {\sqbinom{W}{\ell}}}{h(L)}.
   \]
   We first note all components of degree more than $1$ vanish under $\mathrm{S}$. Indeed, defining $v_W(L) = 1_{L\subseteq W}$, for any $i>1$ we have that 
   \[
   \mathrm{S} P_{=i} h(W)
   =
   \frac{\inner{P_{=i} h}{v_W}}{\Prob{L}{L\subseteq W}}
   =0,
   \]
   where in the last transition we used the fact that $v_W\in J_{\leq 1}$. It follows that 
   \begin{equation}\label{eq:6}
   \mathrm{S} h = \E[h] + \mathrm{S} P_{=1}h.
   \end{equation}
   
   Consider a character $\alpha = u\otimes v$ for $u\in\mathbb{F}_2^n$
   and $v\in\mathbb{F}_2^{\ell}$ non-zero. Then
   \[
   \widehat{f^{\star}}(u\otimes v)
   =
   c
   \Expect{L\in {\sqbinom{\mathbb{F}_2^n}{\ell}}}{f(L)\Expect{M:\mathsf{colspan}(M)=L}{\chi_{u\otimes v}(M)}}.
   \]
   By symmetry, the inner expectation is equal to
   $\Expect{M:\mathsf{colspan}(M)=L}{(-1)^{\inner{u}{\mathsf{col}_1(M)}}}$. As $\mathsf{col}_1(M)$ is a random non-zero vector from $L$, we get that the expectation is $1$ if $L\subseteq \mathsf{span}(u)^{\perp}$, and otherwise is $-\frac{1}{2^{\ell}-1}$. Thus,
   \begin{align*}
   \widehat{f^{\star}}(u\otimes v)
   &= c
   \Expect{L\in {\sqbinom{\mathbb{F}_2^n}{\ell}}}{f(L)
   \left(\frac{2^{\ell}}{2^{\ell}-1}1_{L\subseteq \mathsf{span}(u)^{\perp}}-\frac{1}{2^{\ell}-1}\right)}\\
   &=c\left(
   \frac{2^{\ell}}{2^{\ell}-1}
   \frac{2^{n-\ell}-1}{2^{n}-1}\mathrm{S}f(\mathsf{span}(u)^\perp)-\frac{1}{2^{\ell}-1}\E[f]\right).
   \end{align*}
  Using~\eqref{eq:6} we get
  \begin{align}\label{eq:9}
   \widehat{f^{\star}}(u\otimes v)
   &=
   c\left(
   \frac{2^{\ell}}{2^{\ell}-1}
   \frac{2^{n-\ell}-1}{2^{n}-1}\mathrm{S}P_{=1}f(\mathsf{span}(u)^\perp)+
   \frac{1}{2^{\ell}-1}
   \left(
   \frac{2^{\ell}(2^{n-\ell}-1)}{2^{n}-1}-1\right)\E[f]\right)\notag\\
   &=c\left(
   \frac{2^{\ell}}{2^{\ell}-1}
   \frac{2^{n-\ell}-1}{2^{n}-1}\mathrm{S}P_{=1}f(\mathsf{span}(u)^\perp)-
   \frac{1}{2^{n}-1}\E[f]\right)
  \end{align}
  Re-arranging, squaring and taking sum over $u,v$ gives
  \begin{align}\label{eq:7}
  (2^n-1)\norm{\mathrm{S} P_{=1} f}_2^2=\sum\limits_{u\in\mathbb{F}_2^n\setminus\{0\}}
  \mathrm{S}P_{=1}f(\mathsf{span}(u)^\perp)^2
  &\leq 
  \frac{2}{c}2^{2\ell}\sum\limits_{u\in \mathbb{F}_2^n\setminus\{0\}}2\widehat{f^{\star}}(u\otimes v)^2
  +3\left(\frac{2^{\ell}}{2^{n}-1}\E[f]\right)^2\notag\\
  &\leq \frac{2}{c}2^{\ell}\norm{P_{=1} f^{\star}}_2^2
  +\frac{4}{c}2^{2\ell-n}\E[f]^2,
  \end{align}
  in the last equality we used the fact $f^{\star}$ implies that \[
  \sum\limits_{u\in\mathbb{F}_2^n\setminus\{0\}}2\widehat{f^{\star}}(u\otimes v)^2=\frac{1}{2^{\ell}-1}\sum\limits_{u\in\mathbb{F}_2^n\setminus\{0\},v'\in\mathbb{F}_2^{\ell}\setminus\{0\}}2\widehat{f^{\star}}(u\otimes v')^2,
  \]
  which follows by the fact that $f^{\star}$ is basis invariant, and then bounded the last sum by $\norm{P_{=1}f^{\star}}_2^2$ via Parseval's equality. It remains to relate 
  $\norm{\mathrm{S} P_{=1} f}_2^2$ to 
  $\norm{P_{=1} f}_2^2$. As $P_{=1} f$ is degree $1$ and is orthogonal to the all $1$ function, setting $v_x(L) = 1_{x\in L}$ we may write
  \[
    P_{=1}f(L) = \sum\limits_{x\in\mathbb{F}_2^n\setminus\{0\}}\alpha_x v_x(L),
  \]
  where $\sum\limits_{x}\alpha_x=0$. Note that $\mathrm{S} v_x(W) = \frac{2^{\ell}-1}{2^{n-1}-1}1_{x\in W}$, so
  $\mathrm{S} P_{=1}f(W) = \sum\limits_{x} \alpha_x\frac{2^{\ell}-1}{2^{n-1}-1}1_{x\in W}$ and thus
  \[
  \norm{\mathrm{S} P_{=1}f}_2^2
  =\sum\limits_{x}\alpha_x^2\left(\frac{2^{\ell}-1}{2^{n-1}-1}\right)^2\frac{2^{n-1}-1}{2^{n}-1}
  +\sum\limits_{x\neq y}
  \alpha_x\alpha_y
  \left(\frac{2^{\ell}-1}{2^{n-1}-1}\right)^2\frac{(2^{n-1}-1)(2^{n-1}-2)}{(2^{n}-1)(2^n-2)}.
  \]
  Completing the second sum with $x=y$ and using $\sum_{x}\alpha_x=0$, we get
  \[
  \norm{\mathrm{S} P_{=1}f}_2^2
  =\sum\limits_{x}\left(\frac{2^{\ell}-1}{2^{n-1}-1}\right)^2\alpha_x^2
  \left(\frac{2^{n-1}-1}{2^{n}-1}
  -
  \frac{(2^{n-1}-1)(2^{n-1}-2)}{(2^{n}-1)(2^n-2)}\right)
  =\Theta(2^{2(\ell-n)})
  \sum\limits_{x}\alpha_x^2.
  \]
  Similar calculations give
  \[
  \norm{P_{=1}f}_2^2
  =\sum\limits_{x}\alpha_x^2\frac{2^{\ell}-1}{2^n-1}
  +\sum\limits_{x\neq y}\alpha_x\alpha_y
  \frac{(2^{\ell}-1)(2^{\ell}-2)}{(2^n-1)(2^n-2)}
  =\sum\limits_{x}\alpha_x^2\frac{2^{\ell}-1}{2^n-1}
  \left(1-\frac{2^{\ell}-2}{2^{n}-2}\right),
  \]
  so $\norm{P_{=1}f}_2^2=\Theta(2^{\ell-n})\sum\limits_{x}\alpha_x^2$, and overall we get that 
  $\norm{\mathrm{S}P_{=1}f}_2^2=\Theta(2^{\ell-n})\norm{P_{=1}f}_2^2$. Plugging this into~\eqref{eq:7} gives
  \[
  \Theta(2^{\ell})\norm{P_{=1}f}_2^2\leq \frac{2}{c}2^{\ell}\norm{P_{=1} f^{\star}}_2^2
  +\frac{4}{c}2^{2\ell-n}\E[f]^2,
  \]
  and dividing by $2^{\ell}$ gives~\eqref{eq:8}, thereby completing the proof.

\end{document}